\PassOptionsToPackage{noend}{algorithm2e}
\documentclass{opt2020} 
\pdfoutput=1
\usepackage[utf8]{inputenc} 
\usepackage[T1]{fontenc}    
\usepackage{hyperref}       
\usepackage{booktabs}       
\usepackage{nicefrac}       
\usepackage{microtype}      
\usepackage{epsf}
\usepackage{mathtools}
\usepackage{multirow}
\usepackage{multicol}
\usepackage{xspace}
\usepackage{algorithm,algorithmicx,algpseudocode}
\usepackage{listings}
\usepackage{xcolor}

\usepackage{tikz}
\usetikzlibrary{calc}
\usepackage{zref-savepos}

\usepackage{enumitem}

\usepackage{caption}
\title[Asynchronous Federated Optimization]{Asynchronous Federated Optimization}

\optauthor{\Name{Cong Xie} \Email{cx2@illinois.edu}\\
  \Name{Oluwasanmi Koyejo} \Email{sanmi@illinois.edu}\\
  \Name{Indranil Gupta} \Email{indy@illinois.edu}\\
  \addr Department of Computer Science, University of Illinois Urbana-Champaign}

\begin{document}

\allowdisplaybreaks

\maketitle

\def\Blue{\color{blue}}
\def\Purple{\color{purple}}

\def\A{{\bf A}}
\def\a{{\bf a}}
\def\B{{\bf B}}
\def\C{{\bf C}}
\def\c{{\bf c}}
\def\D{{\bf D}}
\def\d{{\bf d}}
\def\E{{\mathbb{E}}}
\def\F{{\bf F}}
\def\e{{\bf e}}
\def\f{{\bf f}}
\def\G{{\bf G}}
\def\H{{\bf H}}
\def\I{{\bf I}}
\def\K{{\bf K}}
\def\L{{\bf L}}
\def\M{{\bf M}}
\def\m{{\bf m}}
\def\N{{\bf N}}
\def\n{{\bf n}}
\def\Q{{\bf Q}}
\def\q{{\bf q}}
\def\R{{\mathbb{R}}}
\def\S{{\bf S}}
\def\s{{\bf s}}
\def\T{{\bf T}}
\def\U{{\bf U}}
\def\u{{\bf u}}
\def\V{{\bf V}}
\def\v{{\bf v}}
\def\W{{\bf W}}
\def\w{{\bf w}}
\def\X{{\bf X}}
\def\x{{\bf x}}
\def\bx{{\bar{x}}}
\def\Y{{\bf Y}}
\def\y{{\bf y}}
\def\Z{{\bf Z}}
\def\z{{\bf z}}
\def\0{{\bf 0}}
\def\1{{\bf 1}}

\def\AM{{\mathcal A}}
\def\CM{{\mathcal C}}
\def\DM{{\mathcal D}}
\def\GM{{\mathcal G}}
\def\FM{{\mathcal F}}
\def\IM{{\mathcal I}}
\def\NM{{\mathcal N}}
\def\OM{{\mathcal O}}
\def\SM{{\mathcal S}}
\def\TM{{\mathcal T}}
\def\UM{{\mathcal U}}
\def\XM{{\mathcal X}}
\def\YM{{\mathcal Y}}
\def\RB{{\mathbb R}}

\def\TX{\tilde{\bf X}}
\def\tx{\tilde{\bf x}}
\def\ty{\tilde{\bf y}}
\def\TZ{\tilde{\bf Z}}
\def\tz{\tilde{\bf z}}
\def\hd{\hat{d}}
\def\HD{\hat{\bf D}}
\def\hx{\hat{\bf x}}
\def\TD{\tilde{\Delta}}
\def\tg{\tilde{g}}
\def\tmu{\tilde{\mu}}

\def\alp{\mbox{\boldmath$\alpha$\unboldmath}}
\def\bet{\mbox{\boldmath$\beta$\unboldmath}}
\def\epsi{\mbox{\boldmath$\epsilon$\unboldmath}}
\def\etab{\mbox{\boldmath$\eta$\unboldmath}}
\def\ph{\mbox{\boldmath$\phi$\unboldmath}}
\def\pii{\mbox{\boldmath$\pi$\unboldmath}}
\def\Ph{\mbox{\boldmath$\Phi$\unboldmath}}
\def\Ps{\mbox{\boldmath$\Psi$\unboldmath}}
\def\tha{\mbox{\boldmath$\theta$\unboldmath}}
\def\Tha{\mbox{\boldmath$\Theta$\unboldmath}}
\def\muu{\mbox{\boldmath$\mu$\unboldmath}}
\def\Si{\mbox{\boldmath$\Sigma$\unboldmath}}
\def\si{\mbox{\boldmath$\sigma$\unboldmath}}
\def\Gam{\mbox{\boldmath$\Gamma$\unboldmath}}
\def\Lam{\mbox{\boldmath$\Lambda$\unboldmath}}
\def\De{\mbox{\boldmath$\Delta$\unboldmath}}
\def\Ome{\mbox{\boldmath$\Omega$\unboldmath}}
\def\TOme{\mbox{\boldmath$\hat{\Omega}$\unboldmath}}
\def\vps{\mbox{\boldmath$\varepsilon$\unboldmath}}
\newcommand{\ti}[1]{\tilde{#1}}
\def\Ncal{\mathcal{N}}
\def\argmax{\mathop{\rm argmax}}
\def\argmin{\mathop{\rm argmin}}
\providecommand{\abs}[1]{\lvert#1\rvert}
\providecommand{\norm}[2]{\lVert#1\rVert_{#2}}

\def\Zs{{\Z_{\mathrm{S}}}}
\def\Zl{{\Z_{\mathrm{L}}}}
\def\Yr{{\Y_{\mathrm{R}}}}
\def\Yg{{\Y_{\mathrm{G}}}}
\def\Yb{{\Y_{\mathrm{B}}}}
\def\Ar{{\A_{\mathrm{R}}}}
\def\Ag{{\A_{\mathrm{G}}}}
\def\Ab{{\A_{\mathrm{B}}}}
\def\As{{\A_{\mathrm{S}}}}
\def\Asr{{\A_{\mathrm{S}_{\mathrm{R}}}}}
\def\Asg{{\A_{\mathrm{S}_{\mathrm{G}}}}}
\def\Asb{{\A_{\mathrm{S}_{\mathrm{B}}}}}
\def\Or{{\Ome_{\mathrm{R}}}}
\def\Og{{\Ome_{\mathrm{G}}}}
\def\Ob{{\Ome_{\mathrm{B}}}}

\def\Expect{\mathbb{E}}

\def\vec{\mathrm{vec}}
\def\fold{\mathrm{fold}}
\def\index{\mathrm{index}}
\def\sgn{\mathrm{sgn}}
\def\tr{\mathrm{tr}}
\def\rk{\mathrm{rank}}
\def\diag{\mathsf{diag}}
\def\const{\mathrm{Const}}
\def\dg{\mathsf{dg}}
\def\st{\mathsf{s.t.}}
\def\vect{\mathsf{vec}}
\def\MCAR{\mathrm{MCAR}}
\def\MSAR{\mathrm{MSAR}}
\def\etal{{\em et al.\/}\,}
\def\prox{\mathrm{prox}^h_\gamma}
\newcommand{\indep}{{\;\bot\!\!\!\!\!\!\bot\;}}

\newcommand{\mtrxt}[1]{{#1}^\top}
\newcommand{\mtrx}[4]{\left[\begin{matrix}#1 & #2 \\ #3 & #4\end{matrix}\right]}
\DeclarePairedDelimiter\vnorm{\lVert}{\rVert}
\DeclarePairedDelimiterX{\innerprod}[2]{\langle}{\rangle}{#1, #2}

\def\Lsize{\hbox{\space \raise-2mm\hbox{$\textstyle \L \atop \scriptstyle {m\times 3n}$} \space}}
\def\Ssize{\hbox{\space \raise-2mm\hbox{$\textstyle \S \atop \scriptstyle {m\times 3n}$} \space}}
\def\Osize{\hbox{\space \raise-2mm\hbox{$\textstyle \Ome \atop \scriptstyle {m\times 3n}$} \space}}
\def\Tsize{\hbox{\space \raise-2mm\hbox{$\textstyle \T \atop \scriptstyle {3n\times n}$} \space}}
\def\Bsize{\hbox{\space \raise-2mm\hbox{$\textstyle \B \atop \scriptstyle {m\times n}$} \space}}

\newcommand{\twopartdef}[4]
{
	\left\{
		\begin{array}{ll}
			#1 & \mbox{if } #2 \\
			#3 & \mbox{if } #4
		\end{array}
	\right.
}

\newcommand{\tabincell}[2]{\begin{tabular}{@{}#1@{}}#2\end{tabular}}

\newcommand{\NewProcedure}[1]{\vspace{0.3cm}\STATE\hspace{-\algorithmicindent} {\large\underline{\textbf{{#1}:}}}
\setcounter{ALC@line}{0}}
\newcommand{\NewThread}[1]{\vspace{0.2cm}\STATE\hspace{-\algorithmicindent} {\quad\large\underline{\textbf{{#1}:}}} 
\setcounter{ALC@line}{0}}

\DeclarePairedDelimiter\ceil{\lceil}{\rceil}
\DeclarePairedDelimiter\floor{\lfloor}{\rfloor}

\newcommand{\ip}[2]{\left\langle #1, #2 \right \rangle}

\newtheorem{assumption}{Assumption}

\newcommand{\pfcomment}[1]{
\tag*{$\triangleright$ #1}
}

\newcommand{\mcir}[1]{{\large \textcircled{\small #1}}}

\begin{abstract}%
Federated learning enables training on a massive number of edge devices. To improve flexibility and scalability, we propose a new asynchronous federated optimization algorithm. We prove that the proposed approach has near-linear convergence to a global optimum, for both strongly convex and a restricted family of non-convex problems. Empirical results show that the proposed algorithm converges quickly and tolerates staleness in various applications. 
\end{abstract}

\section{Introduction}

Federated learning~\cite{konevcny2016federated,mcmahan2016communication} enables training a global model on datasets partitioned across a massive number of resource-weak edge devices. 
Motivated by the modern phenomenon of distributed (often personal) data collected by edge devices at scale, federated learning can use the large amounts of training data from diverse users for better representation and generalization.
Federated learning is also motivated by the desire for privacy preservation~\cite{bonawitz2019towards,bonawitz2017practical}. 
In some scenarios, on-device training without depositing data in the cloud may be legally required by regulations~\cite{act1996health,FERPA,GDPR}. 

A federated learning system is often composed of servers and workers, with an architecture that is similar to parameter servers~\cite{li2014scaling,li2014communication,ho2013more}. The workers (edge devices) train the models locally on private data. The servers aggregate the learned models from the workers and update the global model.

Federated learning has three key properties~\cite{konevcny2016federated,mcmahan2016communication}: 
1) Infrequent task activation. For the weak edge devices, learning tasks are executed only when the devices are idle, charging, and connected to unmetered networks~\cite{bonawitz2019towards}. 
2) Infrequent communication. The connection between edge devices and the remote servers may frequently be unavailable, slow, or expensive (in terms of communication costs or battery power usage).
3) Non-IID training data. For federated learning, the data on different devices are disjoint, thus may represent non-identically distributed samples from the population. 

Federated learning~\cite{mcmahan2016communication,bonawitz2019towards} is most often implemented using the synchronous approach, which could be slow due to stragglers.
When handling massive edge devices, there could be a large number of stragglers. 
As availability and completion time vary from device to device, due to limited computational capacity and battery time, the global synchronization is difficult, especially in the federated learning scenario. 

Asynchronous training~\cite{zinkevich2009slow,lian2017asynchronous,zheng2017asynchronous} is widely used in traditional distributed stochastic gradient descent (SGD) for stragglers and heterogeneous latency ~\cite{zinkevich2009slow,lian2017asynchronous,zheng2017asynchronous}. 
In this paper, we take the advantage of asynchronous training and combines it with federated optimization.

We propose a novel asynchronous algorithm for federated optimization. The key ideas are (i) to solve regularized local problems to guarantee convergence, and (ii) then use a weighted average to update the global model, where the mixing weight is set adaptively as a function of the staleness. Together, these techniques result in an effective asynchronous federated optimization procedure.
The main contributions of our paper are listed as follows:
\setitemize[0]{leftmargin=*}
\begin{itemize}[topsep=2pt,itemsep=0pt,partopsep=0pt,parsep=0pt]
\item We propose a new asynchronous federated optimization algorithm and a prototype system design.
\item We prove the convergence of the proposed approach for a restricted family of non-convex problems.
\item We propose strategies for controlling the error caused by asynchrony. To this end, we introduce a mixing hyperparameter which adaptively controls the trade-off between the convergence rate and variance reduction according to the staleness.
\item We show empirically that the proposed algorithm converges quickly and often outperforms synchronous federated optimization in practical settings. 
\end{itemize}



\vspace{-0.5cm}
\noindent
\begin{minipage}[t]{0.47\textwidth}
\begin{table}[H]
\caption{Notations and Terminologies.}\vspace{-0.3cm}
\label{tbl:notations}
\begin{center}
\begingroup
\small
\renewcommand{\arraystretch}{0.5}
\setlength{\tabcolsep}{3pt}
\begin{tabular}{|l|l|}
\hline 
{\bf Notation}  & {\bf Description} \\ \hline
$n$    & Number of devices \\ \hline
$T$    & Number of global epochs  \\ \hline
$[n]$    & Set of integers $\{1, \ldots, n \}$  \\ \hline
$H_{min}$    & Minimal number of local iterations  \\ \hline
$H_{max}$  & Maximal number of local iterations  \\ \hline
$\delta$    & $\delta = \frac{H_{max}}{H_{min}}$ is the imbalance ratio  \\ \hline
$H^i_\tau$    & Number of local iterations \\ & in the $\tau^{\mbox{th}}$ epoch  on the $i$th device \\ \hline
$x_t$    & Global model in the $t^{\mbox{th}}$ epoch on server  \\ \hline
$x^i_{\tau,h}$    & Model initialized from $x_\tau$, updated in \\ & the $h$th local iteration, on the $i$th device  \\ \hline
$\mathcal{D}^i$    & Dataset on the $i$th device  \\ \hline
$z^i_{t, h}$    & Data~(minibatch) sampled from $\mathcal{D}^i$  \\ \hline
$\gamma$    & Learning rate  \\ \hline
$\alpha$    & Mixing hyperparameter  \\ \hline
$\rho$    & Regularization weight  \\ \hline
$t-\tau$    & Staleness  \\ \hline
$s(t-\tau)$    & Function of staleness for adaptive $\alpha$  \\ \hline
$\| \cdot \|$    & All the norms in this paper are $l_2$-norms  \\ \hline
Device    & Where the training data are placed  \\ \hline
Worker    & One worker on each device, \\ & process that trains the model  \\ \hline
\end{tabular}
\endgroup
\end{center}
\end{table}
\end{minipage}
\hfill
\begin{minipage}[t]{0.510\textwidth}
\begin{figure}[H]
\centering
\includegraphics[width=\textwidth,height=4.1cm]{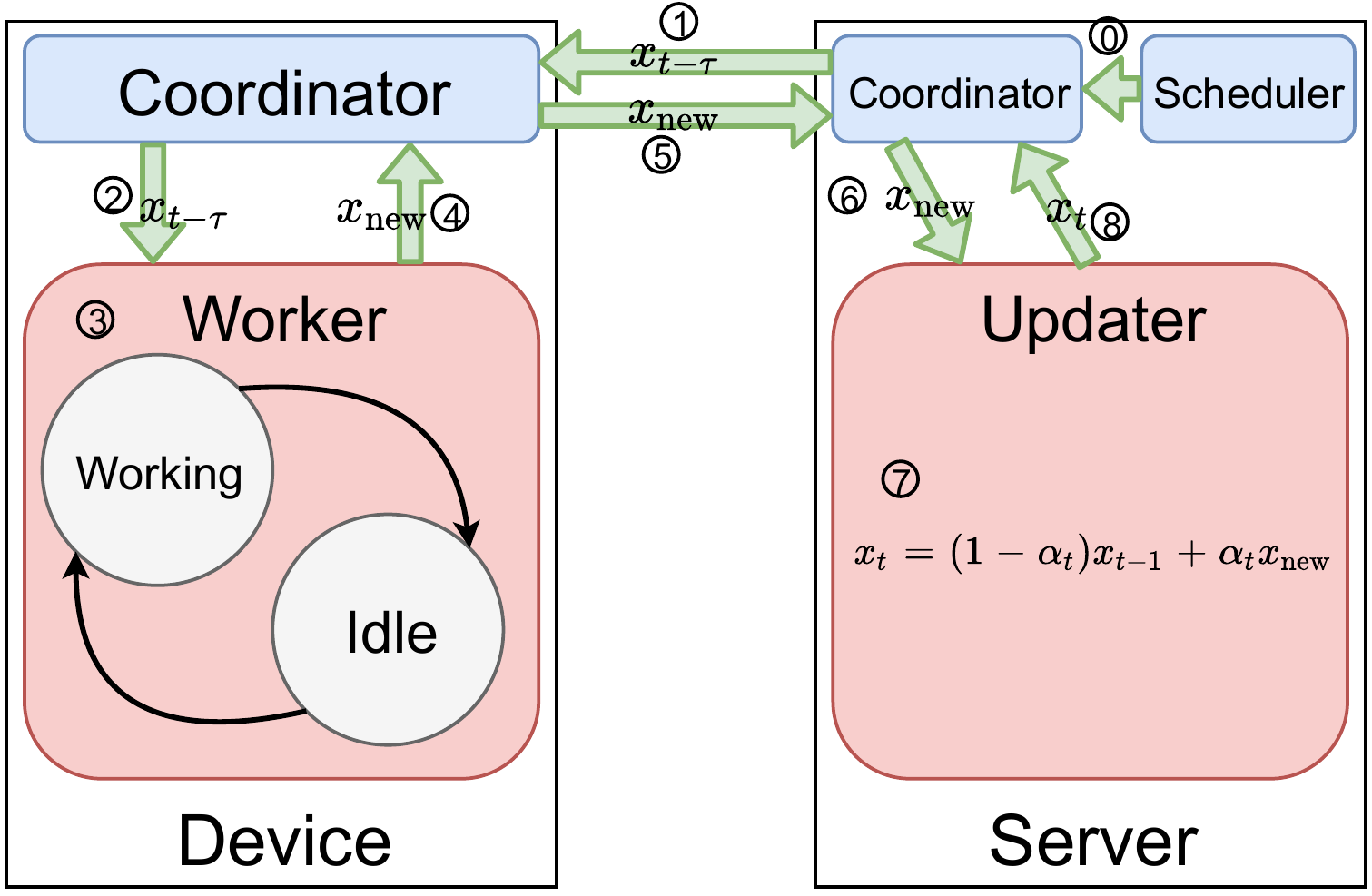}\vspace{-0.2cm}
\caption{System overview. \mcir{0}: scheduler triggers training through coordinator. \mcir{1}, \mcir{2}: worker receives model $x_{t-\tau}$ from server via coordinator. \mcir{3}: worker computes local updates as Algorithm~\ref{alg:fed_async}. Worker can switch between the two states: working and idle. \mcir{4}, \mcir{5}, \mcir{6}: worker pushes the locally updated model to server via the coordinator. Coordinator queues the models received in \mcir{5}, and feeds them to the updater sequentially in \mcir{6}. \mcir{7}, \mcir{8}: server updates the global model and makes it ready to read in the coordinator. In our system, \mcir{1} and \mcir{5} operate asynchronously in parallel.}
\label{fig:overview}
\end{figure}
\end{minipage}

\vspace{-0.5cm}
\section{Problem formulation}
We consider federated learning with $n$ devices. On each device, a worker process trains a model on local data. The overall goal is to train a global model $x \in \R^d$ using data from all the devices. Formally, we solve
$
\min_{x \in \R^d} F(x), 
$
where $F(x) = \frac{1}{n} \sum_{i \in [n]} \E_{z^i \sim \mathcal{D}^i} f(x; z^i)$, for $\forall i \in [n]$, $z^i$ is sampled from the local data $\mathcal{D}^i$ on the $i$th device.
Note that different devices have different local datasets, i.e., $\mathcal{D}^i \neq \mathcal{D}^j, \forall i \neq j$. 

\section{Methodology}
\label{sec:methodology}

The training takes $T$ global epochs. In the $t^{\mbox{th}}$ epoch, the server receives a locally trained model $x_{new}$ from an arbitrary worker, and updates the global model by weighted averaging:
$
x_t = (1-\alpha) x_{t-1} + \alpha x_{new},
$
where $\alpha \in (0, 1)$ is the mixing hyperparameter.
A system overview is illustrated in Figure~\ref{fig:overview}.

On an arbitrary device $i$, after receiving a global model $x_t$ (potentially stale) from the server, we locally solve the following regularized optimization problem using SGD for multiple iterations:
$
\min_{x \in \R^d} \E_{z^i \sim \mathcal{D}^i} f(x; z^i) + \frac{\rho}{2} \| x - x_{t} \|^2.
$
For convenience, we define $g_{x'}(x;z) = f(x; z) + \frac{\rho}{2}\|x - x'\|^2$.

\begin{algorithm}[hbt!]
\caption{Asynchronous Federated Optimization (FedAsync)}
\SetKwFunction{FServer}{Server}
\SetKwFunction{FScheduler}{Scheduler}
\SetKwFunction{FUpdater}{Updater}
\SetKwFunction{FWorker}{Worker}

\SetKwProg{Fn}{Process}{:}{}
\Fn{\FServer{$\alpha \in (0,1)$}}{
Initialize $x_0$, $\alpha_t \leftarrow \alpha, \forall t \in [T]$ \\
Run Scheduler() thread and Updater() thread asynchronously in parallel
}

\SetKwProg{Fn}{Thread}{:}{}
\Fn{\FScheduler{}}{
Periodically trigger training tasks on some workers, and send the global model with time stamp
}

\SetKwProg{Fn}{Thread}{:}{}
\Fn{\FUpdater{}}{
\For{epoch $t \in [T]$}{
	Receive the pair $(x_{new}, \tau)$ from any worker \\
	Optional: $\alpha_t \leftarrow \alpha \times s(t - \tau)$, $s(\cdot)$ is a function of the staleness \\
	$
	x_t \leftarrow (1-\alpha_t) x_{t-1} + \alpha_t x_{new}
	$
}
}

\SetKwProg{Fn}{Process}{:}{}
\Fn{\FWorker{}}{
\For{$i \in [n]$ in parallel}{
	\If{triggered by the scheduler}{
	Receive the pair of the global model and its time stamp $(x_t, t)$ from the server \\
	$\tau \leftarrow t$, $x_{\tau, 0}^i \leftarrow x_t$ \\
	Define $g_{x_t}(x;z) = f(x; z) + \frac{\rho}{2}\|x - x_t\|^2$, where $\rho > \mu$ \\
	\For{local iteration $h \in [H_\tau^i]$}{
		Randomly sample $z_{\tau, h}^i \sim \mathcal{D}^i$ \\
		Update 
		$
		x_{\tau, h}^i \leftarrow x_{\tau, h-1}^i - \gamma \nabla g_{x_t}(x_{\tau, h-1}^i; z_{\tau, h}^i)
		$
	}
	Push $(x_{\tau, H_\tau^i}^i, \tau)$ to the server 
	}
}
}

\label{alg:fed_async}
\end{algorithm} 

The server and workers conduct updates asynchronously, i.e., the server immediately updates the global model whenever it receives a local model. The communication between the server and the workers is non-blocking. 
Thus, the server and workers can update the models at any time without synchronization, which is favorable when the devices have heterogeneous conditions. 

The detailed algorithm is shown in Algorithm~\ref{alg:fed_async}. The model parameter $x_{\tau, h}^i$ is updated in the $h$th local iteration after receiving $x_\tau$, on the $i$th device. The data $z_{\tau, h}^i$ is randomly drawn in the $h$th local iteration after receiving $x_\tau$, on the $i$th device. $H_{\tau}^i$ is the number of local iterations after receiving $x_\tau$ on the $i$th device. $\gamma$ is the learning rate and $T$ is the total number of global epochs.

\begin{remark}
On the server side, the scheduler and the updater run asynchronously in parallel. The scheduler periodically triggers training tasks and controls the staleness~($t-\tau$ in the updater thread). The updater receives models from workers and updates the global model. Our architecture allows for multiple updater threads with read-write lock on the global model, which improves the throughput. 
\end{remark}

\begin{remark}
Intuitively, larger staleness results in greater error when updating the global model. For the local models with large staleness $(t-\tau)$, we can decrease $\alpha$ to mitigate the error caused by staleness. 
As shown in Algorithm~\ref{alg:fed_async}, optionally, we use a function $s(t-\tau)$ to determine the value of $\alpha$. In general, $s(t-\tau)$ should be $1$ when $t = \tau$, and monotonically decrease when $(t-\tau)$ increases. There are many functions that satisfy such two properties, with different decreasing rate, e.g., $s_a(t-\tau) = \frac{1}{t-\tau+1}$. 
The options used in this paper can be found in Section~\ref{sec:evaluation_setup}.
\end{remark}

\section{Convergence analysis}

First, we introduce some definitions and assumptions for our convergence analysis.
\begin{definition} (Smoothness)
A differentiable function $f$ is $L$-smooth if for $\forall x, y$, 
$
f(y) - f(x) \leq \ip{\nabla f(x)}{y-x} + \frac{L}{2} \|y-x\|^2,
$
where $L > 0$.
\end{definition}

\begin{definition} (Weak convexity)
A differentiable function $f$ is $\mu$-weakly convex if the function $g$ with $g(x) = f(x) + \frac{\mu}{2} \|x\|^2$ is convex,
where $\mu \geq 0$.
$f$ is convex if $\mu=0$, and non-convex if $\mu > 0$.
\end{definition}

We have the following convergence guarantees. Detailed proofs can be found in the appendix. 

\begin{theorem}
Assume that $F$ is $L$-smooth and $\mu$-weakly convex, and each worker executes at least $H_{min}$ and at most $H_{max}$ local updates before pushing models to the server. We assume bounded delay $t-\tau \leq K$. The imbalance ratio of local updates is $\delta = \frac{H_{max}}{H_{min}}$. Furthermore, we assume that for $\forall x \in \R^d, i \in [n]$, and $\forall z \sim \mathcal{D}^i$, we have $\| \nabla f(x; z) \|^2 \leq V_1$ and $\| \nabla g_{x'}(x; z) \|^2 \leq V_2$, $\forall x'$. For any small constant $\epsilon > 0$, taking $\rho$ large enough such that $\rho > \mu$ and $- (1 + 2 \rho + \epsilon) V_2 + \rho^2 \|x_{\tau, h-1} - x_\tau\|^2 - \frac{\rho}{2} \|x_{\tau, h-1} - x_\tau\|^2 \geq 0, \forall x_{\tau, h-1}, x_{\tau}$, and $\gamma < \frac{1}{L}$, after $T$ global updates, Algorithm~\ref{alg:fed_async} converges to a critical point: 
$
\min_{t = 0}^{T-1} \E \| \nabla F(x_t) \|^2 \leq \frac{\E\left[ F(x_{0}) - F(x_T) \right]}{\alpha \gamma \epsilon T H_{min}}  
+ \OM\left( \frac{\gamma H_{max}^3 + \alpha K H_{max} }{\epsilon H_{min}} \right)
+ \OM\left( \frac{\alpha^2\gamma K^2 H_{max}^2 + \gamma K^2 H_{max}^2 }{\epsilon H_{min}} \right).
$
\end{theorem}




\section{Experiments}

In this section, we empirically evaluate the proposed algorithm. 

\subsection{Datasets}
We conduct experiments on two benchmarks: CIFAR-10~\cite{krizhevsky2009learning}, and WikiText-2~\cite{merity2016pointer}. The training set is partitioned onto $n=100$ devices. 
The mini-batch sizes are 50 and 20 respectively.

\subsection{Evaluation setup}
\label{sec:evaluation_setup}


The baseline algorithm is \textit{FedAvg} introduced by \cite{mcmahan2016communication}, which implements synchronous federated optimization. For \textit{FedAvg}, in each epoch, $k=10$ devices are randomly selected to launch local updates. We also consider single-thread SGD as a baseline. For \textit{FedAsync}, we simulate the asynchrony by randomly sampling the staleness $(t-\tau)$ from a uniform distribution. 

We repeat each experiment 10 times and take the average. For CIFAR-10, we use the top-1 accuracy on the testing set as the evaluation metric. 
To compare asynchronous training and synchronous training, we consider ``metrics vs. number of gradients''.
The ``number of gradients'' is the number of gradients applied to the global model. 

For convenience, we name Algorithm~\ref{alg:fed_async} as \textit{FedAsync}. We also test the performance of \textit{FedAsync} with adaptive mixing hyperparameters $\alpha_t = \alpha \times s(t-\tau)$, as outlined in Section~\ref{sec:methodology}. We employ the following three strategies for the weighting function $s(t - \tau)$ (parameterized by $a,b>0$):
\noindent
\begin{minipage}[t]{0.41\textwidth}
\setitemize[0]{leftmargin=*}
\begin{itemize}[topsep=2pt,itemsep=0pt,partopsep=0pt,parsep=0pt]
\item Constant: $s(t-\tau) = 1$.
\item Polynomial: $s_a(t-\tau) = (t-\tau+1)^{-a}$.
\end{itemize}
\end{minipage}
\hfill
\begin{minipage}[t]{0.58\textwidth}
\setitemize[0]{leftmargin=*}
\begin{itemize}[topsep=2pt,itemsep=0pt,partopsep=0pt,parsep=0pt]
\item Hinge: 
$
s_{a,b}(t-\tau) = 
\begin{cases}
1 & \mbox{if $t-\tau \leq b$}\\
\frac{1}{a(t-\tau-b)+1} & \mbox{otherwise}
\end{cases}.
$
\end{itemize}
\end{minipage}

For convenience, we refer to \textit{FedAsync} with constant $\alpha$ as \textit{FedAsync+Const}, \textit{FedAsync} with polynomial adaptive $\alpha$ as \textit{FedAsync+Poly}, and \textit{FedAsync} with hinge adaptive $\alpha$ as \textit{FedAsync+Hinge}.


\begin{figure*}[htb!]
\centering
\subfigure[Top-1 accuracy on testing set, $t-\tau \leq 4$]{\includegraphics[width=0.495\textwidth]{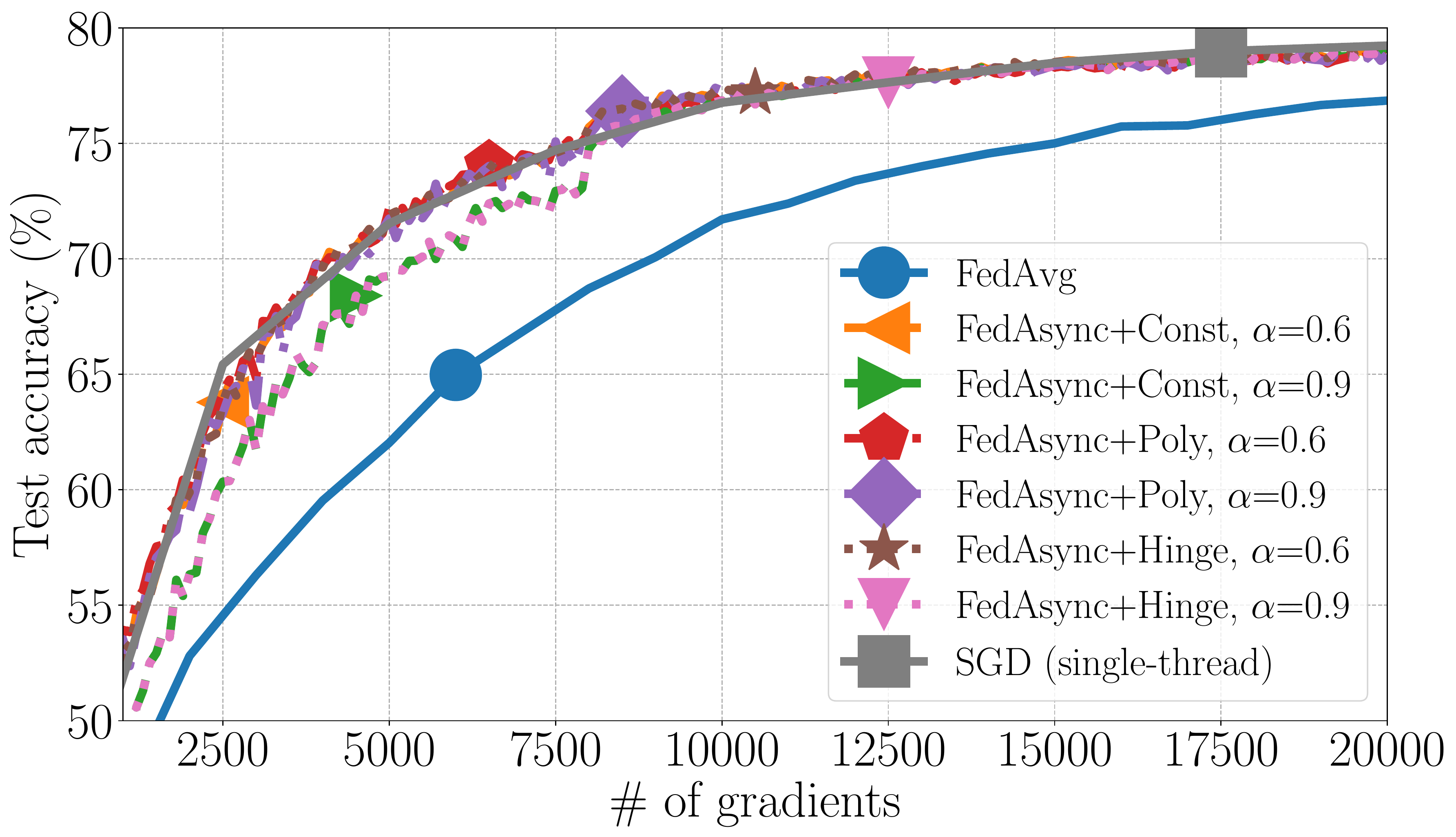}}
\subfigure[Top-1 accuracy on testing set, $t-\tau \leq 16$]{\includegraphics[width=0.495\textwidth]{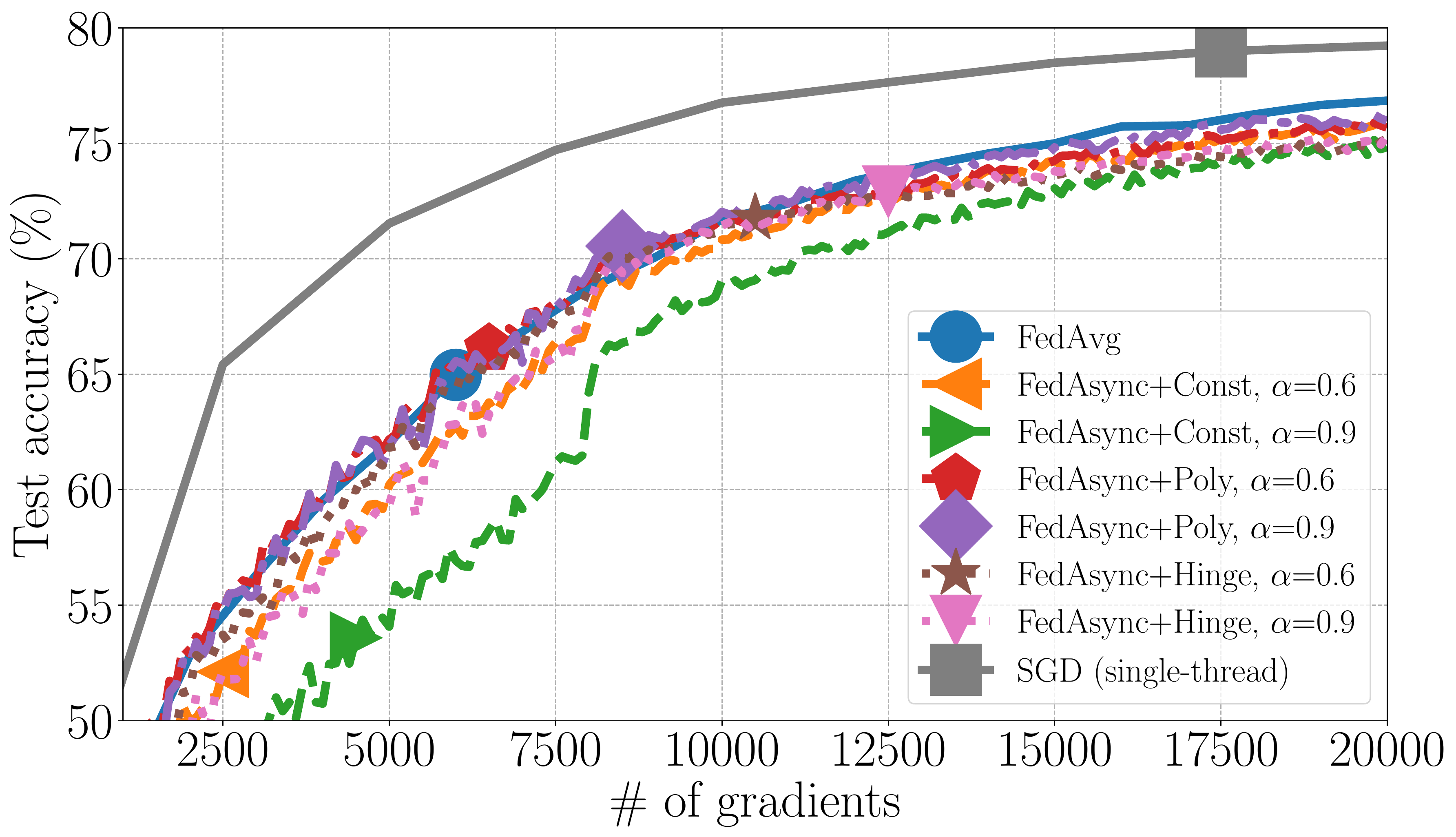}}
\caption{Top-1 accuracy~(the higher the better) vs. \# of gradients on CNN and CIFAR-10 dataset. The maximum staleness is $4$ or $16$. $\gamma = 0.1$, $\rho = 0.005$. For \textit{FedAsync+Poly}, we take $a=0.5$. For \textit{FedAsync+Hinge}, we take $a=10, b=4$. Note that when the maximum staleness is $4$, \textit{FedAsync+Const} and \textit{FedAsync+Hinge} with $b=4$ are the same.}
\label{fig:gradient_cifar10}
\end{figure*}
\vspace{-0.3cm}

\subsection{Empirical results}



We test \textit{FedAsync}~(asynchronous federated optimization in Algorithm~\ref{alg:fed_async}) with different learning rates $\gamma$, regularization weights $\rho$, mixing hyperparameter $\alpha$, and staleness. 


In Figure~\ref{fig:gradient_cifar10} and \ref{fig:gradient_wikitext}, we show how \textit{FedAsync} converges when the number of gradients grows. We can see that when the overall staleness is small, \textit{FedAsync} converges as fast as \textit{SGD}, and faster than \textit{FedAvg}. When the staleness is larger, \textit{FedAsync} converges slower. In the worst case, \textit{FedAsync} has similar convergence rate as \textit{FedAvg}. When  $\alpha$ is too large, the convergence can be unstable, especially for \textit{FedAsync+Const}. The convergence is more robust when adaptive $\alpha$ is used. 

\begin{figure*}[htb!]
\centering
\subfigure[Perplexity on testing set, $t-\tau \leq 4$]{\includegraphics[width=0.495\textwidth]{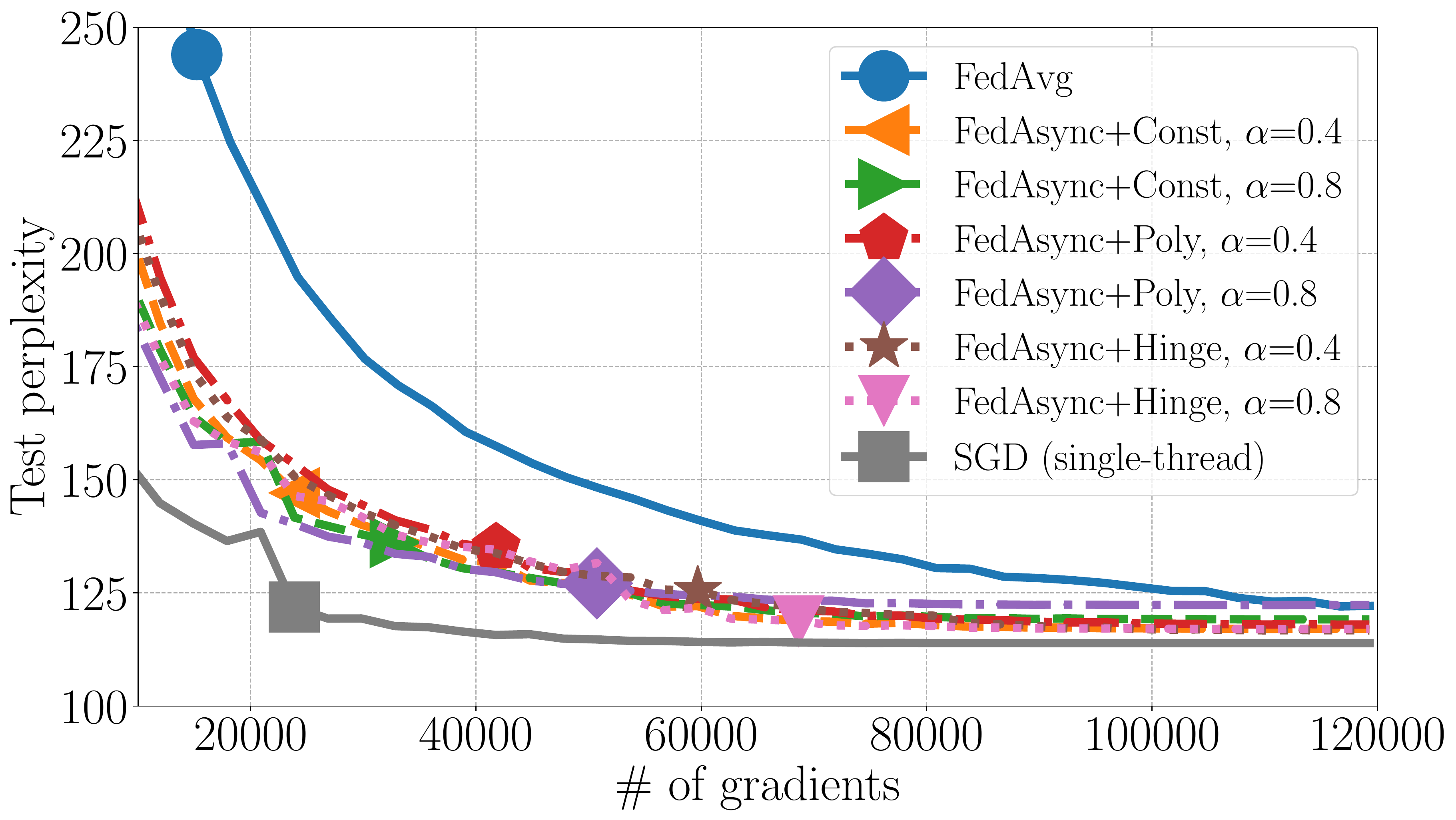}}
\subfigure[Perplexity on testing set, $t-\tau \leq 16$]{\includegraphics[width=0.495\textwidth]{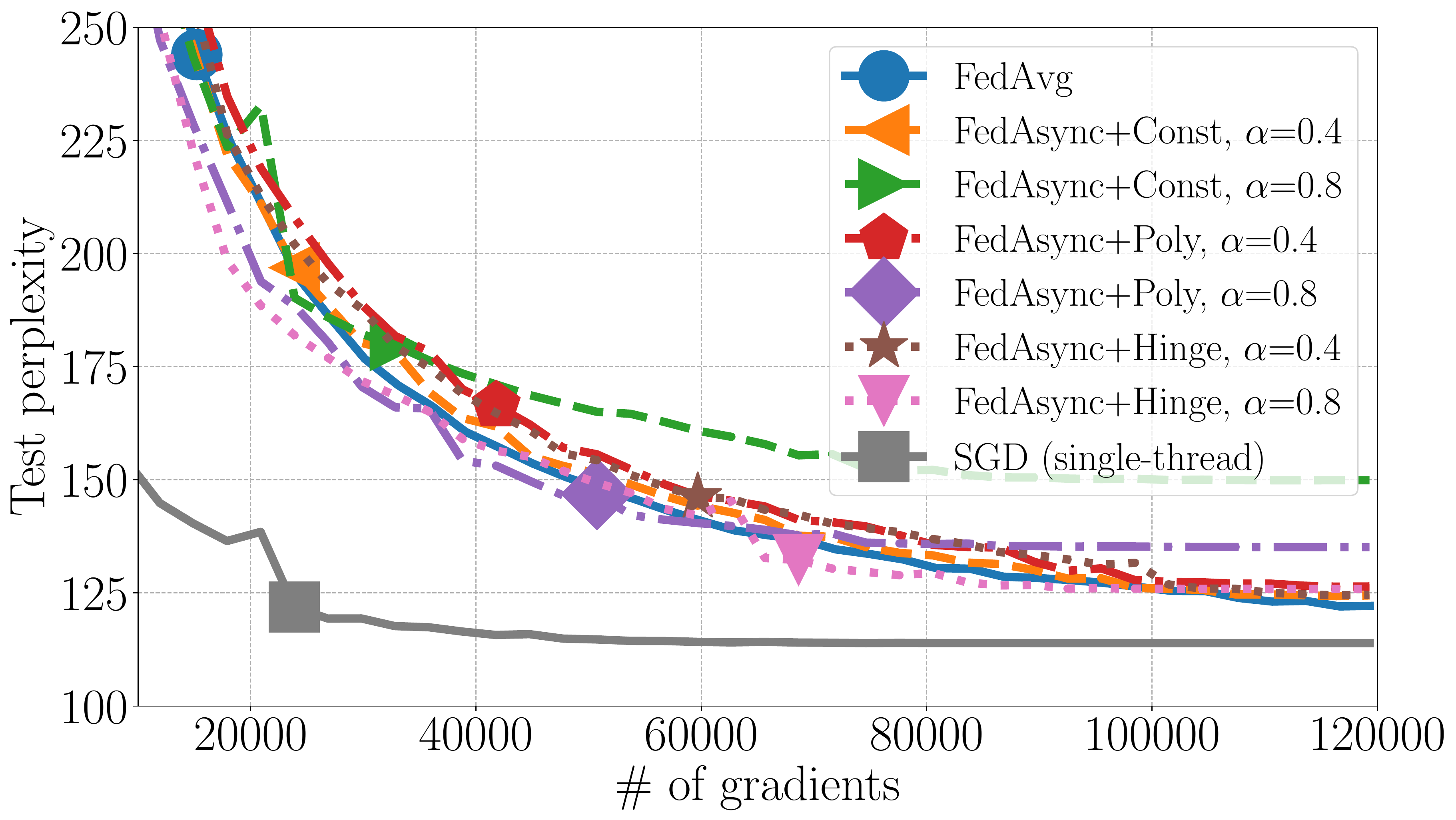}}
\caption{Perplexity~(the lower the better) vs. \# of gradients on LSTM-based language model and WikiText-2 dataset. The maximum staleness is $4$ or $16$. $\gamma = 20$, $\rho = 0.0001$.  For \textit{FedAsync+Poly}, we take $a=0.5$. For \textit{FedAsync+Hinge}, we take $a=10, b=2$.}
\label{fig:gradient_wikitext}
\end{figure*}

In Figure~\ref{fig:delay}, we show how staleness affects the convergence of \textit{FedAsync}, evaluated on CNN and CIFAR-10 dataset. Overall, larger staleness makes the convergence slower, but the influence is not catastrophic. Furthermore, the instability caused by large staleness can be mitigated by using adaptive $\alpha$. Using adaptive $\alpha$ always improves the performance, compared to using constant $\alpha$.

\begin{figure}[htb!]
\centering
\includegraphics[width=0.495\textwidth]{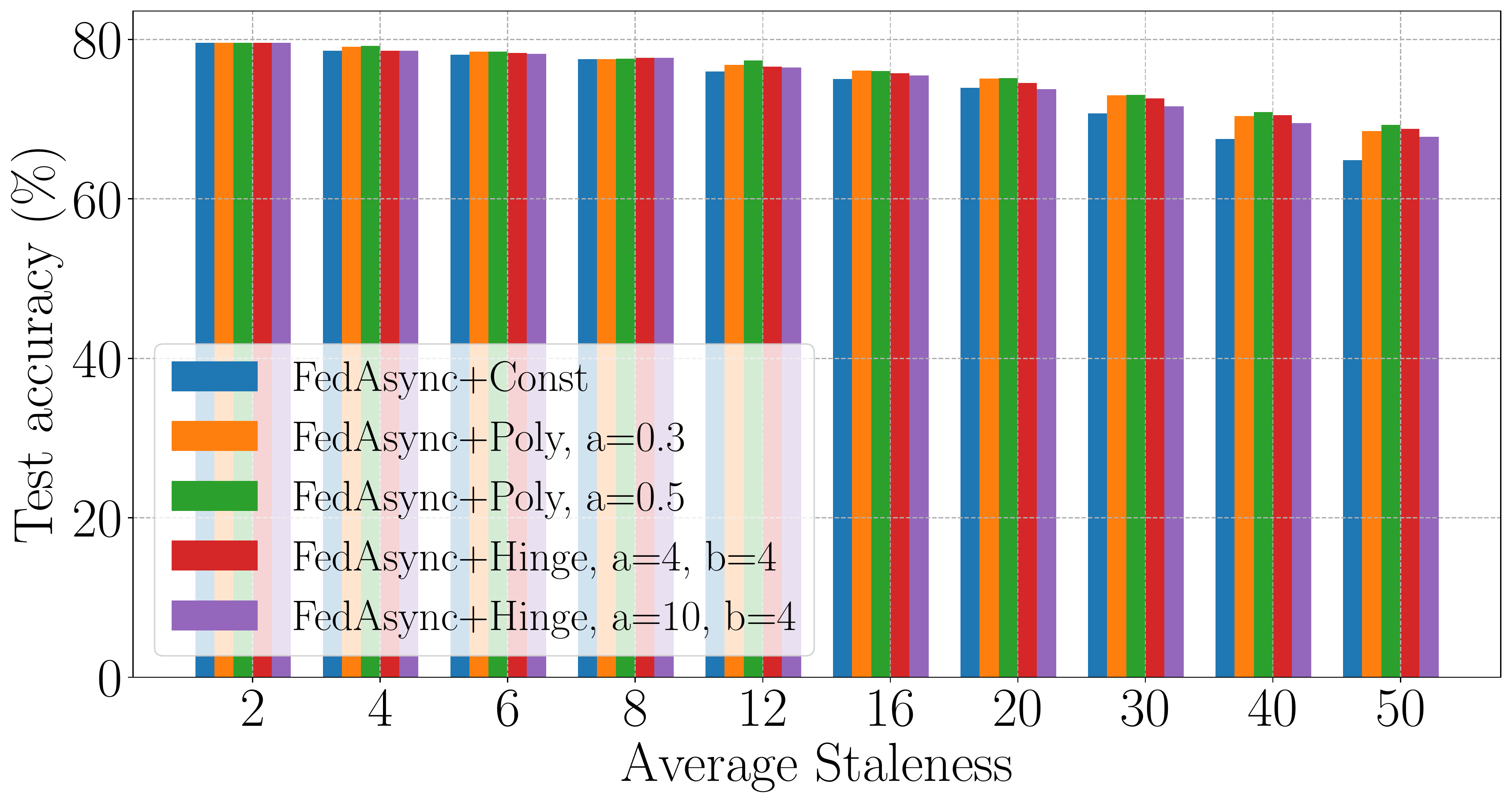}
\caption{Top-1 accuracy on CNN and CIFAR-10 dataset at the end of training, with different staleness. $\gamma = 0.1$, $\rho = 0.01$. $\alpha$ has initial value $0.9$.}
\label{fig:delay}
\end{figure}

\subsection{Discussion}

In general, the convergence rate of \textit{FedAsync} is between single-thread \textit{SGD} and \textit{FedAvg}. Larger $\alpha$ and smaller staleness make \textit{FedAsync} closer to single-thread \textit{SGD}. Smaller $\alpha$ and larger staleness makes \textit{FedAsync} closer to \textit{FedAvg}.

Empirically, we observe that \textit{FedAsync} is generally insensitive to hyperparameters. When the staleness is large, we can tune $\alpha$ to improve the convergence.  Without adaptive $\alpha$, smaller $\alpha$ is better for larger staleness. For adaptive $\alpha$, our best choice empirically was \textit{FedAsync+Hinge}. 
\textit{FedAsync+Poly} and \textit{FedAsync+Hinge} have similar performance. 

In summary, compared to \textit{FedAvg}, \textit{FedAsync} performs as good as, and in most cases better. 
When the staleness is small, \textit{FedAsync} converges much faster than \textit{FedAvg}. When the staleness is large, \textit{FedAsync} still achieves similar performance as \textit{FedAvg}.

\section{Conclusion}

We proposed a novel asynchronous federated optimization algorithm on non-IID training data. We proved the convergence for a restricted family of non-convex problems. Our empirical evaluation validated both fast convergence and staleness tolerance. An interesting future direction is the design of strategies to adaptively tune the mixing hyperparameters. 

\acks{This work was funded in part by the following grants: NSF IIS 1909577, NSF CNS 1908888, NSF CCF 1934986 and a JP Morgan Chase Fellowship, along with computational resources donated by Intel, AWS, and Microsoft Azure.}

\bibliography{fed_async_opt2020}

\newpage
\onecolumn
\appendix
\vspace*{0.1cm}
\begin{center}
	\Large\textbf{Appendix}
\end{center}
\vspace*{0.1cm}

\section{Proofs}

\setcounter{theorem}{0}

\begin{theorem}
Assume that $F$ is $L$-smooth and $\mu$-weakly convex, and each worker executes at least $H_{min}$ and at most $H_{max}$ local updates before pushing models to the server. We assume bounded delay $t-\tau \leq K$. The imbalance ratio of local updates is $\delta = \frac{H_{max}}{H_{min}}$. Furthermore, we assume that for $\forall x \in \R^d, i \in [n]$, and $\forall z \sim \mathcal{D}^i$, we have $\| \nabla f(x; z) \|^2 \leq V_1$ and $\| \nabla g_{x'}(x; z) \|^2 \leq V_2$, $\forall x'$. Taking $\rho$ large enough such that $\rho > \mu$ and $- (1 + 2 \rho + \epsilon) V_2 + \rho^2 \|x_{\tau, h-1} - x_\tau\|^2 - \frac{\rho}{2} \|x_{\tau, h-1} - x_\tau\|^2 \geq 0, \forall x_{\tau, h-1}, x_{\tau}$, and $\gamma < \frac{1}{L}$, after $T$ global updates, Algorithm~\ref{alg:fed_async} converges to a critical point:
\begin{align*}
&\min_{t = 0}^{T-1} \E \| \nabla F(x_t) \|^2  \\
&\leq \frac{\E\left[ F(x_{0}) - F(x_T) \right]}{\alpha \gamma \epsilon T H_{min}}  + \OM\left( \frac{\gamma H_{max}^3 + \alpha K H_{max} + \alpha^2\gamma K^2 H_{max}^2 + \gamma K^2 H_{max}^2 }{\epsilon H_{min}} \right).
\end{align*}
Taking $\alpha = \frac{1}{\sqrt{H_{min}}}$, $\gamma = \frac{1}{\sqrt{T}}$, $T = H_{min}^5$, we have 
\begin{align*}
\min_{t = 0}^{T-1} \E \| \nabla F(x_t) \|^2   
\leq \OM\left( \frac{1}{\epsilon H_{min}^3} 
 + \frac{\delta^3}{\epsilon \sqrt{H_{min}}} 
 + \frac{K \delta}{\epsilon \sqrt{H_{min}}} 
 + \frac{K^2 \delta^2}{\epsilon \sqrt{H_{min}^5}}  
 + \frac{K^2 \delta^2}{\epsilon \sqrt{H_{min}^3}} \right).
\end{align*}
\end{theorem}

\begin{proof}
Without loss of generality, we assume that in the $t^{\mbox{th}}$ epoch, the server receives the model $x_{new}$, with time stamp $\tau$. We assume that $x_{new}$ is the result of applying $H_{min} \leq H \leq H_{max}$ local updates to $x_{\tau}$ on the $i$th device. We also ignore $i$ in $x_{\tau, h}^i$ and $z_{\tau, h}^i$ for convenience.

Thus, using smoothness and strong convexity, conditional on $x_{\tau, h-1}$, for $\forall h \in [H]$ we have
\begin{align*}
&\E\left[ F(x_{\tau, h}) - F(x_*) \right] \\
&\leq \E\left[ G_{x_\tau}(x_{\tau, h}) - F(x_*) \right] \\
&\leq G_{x_\tau}(x_{\tau, h-1}) - F(x_*) - \gamma \E\left[  \ip{\nabla G_{x_\tau}(x_{\tau, h-1})}{\nabla g_{x_\tau}(x_{\tau, h-1}; z_{\tau, h})} \right] \\
&\quad + \frac{L \gamma^2}{2} \E\left[  \| \nabla g_{x_\tau}(x_{\tau, h-1}; z_{\tau, h}) \|^2 \right] \\
&\leq F(x_{\tau, h-1}) - F(x_*) + \frac{\rho}{2} \|x_{\tau, h-1} - x_{\tau}\|^2 - \gamma \E\left[  \ip{\nabla G_{x_\tau}(x_{\tau, h-1})}{\nabla g_{x_\tau}(x_{\tau, h-1}; z_{\tau, h})} \right] \\
&\quad + \frac{L \gamma^2}{2} \E\left[  \| \nabla g_{x_\tau}(x_{\tau, h-1}; z_{\tau, h}) \|^2 \right] \\
&\leq F(x_{\tau, h-1}) - F(x_*) - \gamma \E\left[  \ip{\nabla G_{x_\tau}(x_{\tau, h-1})}{\nabla g_{x_\tau}(x_{\tau, h-1}; z_{\tau, h})} \right] + \frac{L \gamma^2}{2} V_2 + \frac{\rho H_{max}^2 \gamma^2}{2} V_2 \\
&\leq F(x_{\tau, h-1}) - F(x_*) - \gamma \E\left[  \ip{\nabla G_{x_\tau}(x_{\tau, h-1})}{\nabla g_{x_\tau}(x_{\tau, h-1}; z_{\tau, h})} \right] + \gamma^2 \OM(\rho H_{max}^2 V_2).
\end{align*}

Taking $\rho$ large enough such that $- (1 + 2 \rho + \epsilon) V_1 + \rho^2 \|x_{\tau, h-1} - x_\tau\|^2 - \frac{\rho}{2} \|x_{\tau, h-1} - x_\tau\|^2 \geq 0, \forall x_{\tau, h-1}, x_{\tau}$, and write $\nabla g_{x_\tau}(x_{\tau, h-1}; z_{\tau, h})$ as $\nabla g_{x_\tau}(x_{\tau, h-1})$ for convenience, we have 
\begin{align*}
&\ip{\nabla G_{x_\tau}(x_{\tau, h-1})}{\nabla g_{x_\tau}(x_{\tau, h-1})} - \epsilon \|\nabla F(x_{\tau, h-1})\|^2 \\
&= \ip{\nabla F(x_{\tau, h-1}) + \rho (x_{\tau, h-1} - x_\tau)}{\nabla f(x_{\tau, h-1}) + \rho (x_{\tau, h-1} - x_\tau)} - \epsilon \|\nabla F(x_{\tau, h-1})\|^2 \\
&= \ip{\nabla F(x_{\tau, h-1})}{\nabla f(x_{\tau, h-1})} + \rho \ip{\nabla F(x_{\tau, h-1}) + \nabla f(x_{\tau, h-1})}{x_{\tau, h-1} - x_\tau} \\
&\quad + \rho^2 \|x_{\tau, h-1} - x_\tau\|^2 - \epsilon \|\nabla F(x_{\tau, h-1})\|^2 \\
&\geq -\frac{1}{2} \|\nabla F(x_{\tau, h-1})\|^2 - \frac{1}{2} \|\nabla f(x_{\tau, h-1})\|^2 -\frac{\rho}{2} \|\nabla F(x_{\tau, h-1}) + \nabla f(x_{\tau, h-1})\|^2  \\
&\quad -\frac{\rho}{2} \|x_{\tau, h-1} - x_\tau\|^2 + \rho^2 \|x_{\tau, h-1} - x_\tau\|^2 - \epsilon \|\nabla F(x_{\tau, h-1})\|^2 \\
&\geq -\frac{1}{2} \|\nabla F(x_{\tau, h-1})\|^2 - \frac{1}{2} \|\nabla f(x_{\tau, h-1})\|^2 -\rho \|\nabla F(x_{\tau, h-1}) \|^2 - \rho \| \nabla f(x_{\tau, h-1})\|^2  \\
&\quad -\frac{\rho}{2} \|x_{\tau, h-1} - x_\tau\|^2 + \rho^2 \|x_{\tau, h-1} - x_\tau\|^2 - \epsilon \|\nabla F(x_{\tau, h-1})\|^2 \\
&\geq - (1 + 2 \rho + \epsilon) V_1 + \rho^2 \|x_{\tau, h-1} - x_\tau\|^2 - \frac{\rho}{2} \|x_{\tau, h-1} - x_\tau\|^2 \\
&= a \rho^2 + b \rho + c \geq 0,
\end{align*}
where $a = \|x_{\tau, h-1} - x_\tau\|^2 > 0$, $b = -2 V_1 - \frac{1}{2} \|x_{\tau, h-1} - x_\tau\|^2$, $c = - (1 + \epsilon) V_1$.
Thus, we have $\gamma \ip{\nabla G_{x_\tau}(x_{\tau, h-1})}{\nabla g_{x_\tau}(x_{\tau, h-1})} \leq \gamma \epsilon \|\nabla F(x_{\tau, h-1})\|^2$.

Using $\tau - (t-1) \leq K$, we have 
$
\|x_\tau - x_{t-1}\|^2 
\leq \|(x_\tau - x_{\tau+1}) + \ldots + (x_{t-1} - x_{t-1})\|^2 
\leq K\|x_\tau - x_{\tau+1}\|^2 + \ldots + K\|x_{t-1} - x_{t-1}\|^2 
\leq \alpha^2\gamma^2 K^2 H_{max}^2 \OM(V_2).
$

Also, we have
$
\|x_\tau - x_{t-1}\| 
\leq \|(x_\tau - x_{\tau+1}) + \ldots + (x_{t-1} - x_{t-1})\| 
\leq \|x_\tau - x_{\tau+1}\|^2 + \ldots + \|x_{t-1} - x_{t-1}\|^2 
\leq \alpha\gamma K H_{max} \OM(\sqrt{V_2}).
$

Thus, we have
\begin{align*}
&\E\left[ F(x_{\tau, h}) - F(x_*) \right] \\
&\leq F(x_{\tau, h-1}) - F(x_*) - \gamma \E\left[  \ip{\nabla G_{x_\tau}(x_{\tau, h-1})}{\nabla g_{x_\tau}(x_{\tau, h-1}; z_{\tau, h})} \right] + \gamma^2 \OM(\rho H_{max}^2 V_2) \\
&\leq F(x_{\tau, h-1}) - F(x_*) - \gamma \epsilon \|\nabla F(x_{\tau, h-1})\|^2 + \gamma^2 \OM(\rho H_{max}^2 V_2)
\end{align*}

By rearranging the terms and telescoping, we have
\begin{align*}
\E\left[ F(x_{\tau, H}) - F(x_{\tau}) \right]  \leq - \gamma \epsilon \sum_{h = 0}^{H-1} \E \|\nabla F(x_{\tau, h})\|^2 + \gamma^2 \OM(\rho H_{max}^3 V_2).
\end{align*}

Then, we have
\begin{align*}
&\E\left[ F(x_t) - F(x_{t-1}) \right] \\
&\leq \E\left[ G_{x_{t-1}}(x_t) - F(x_{t-1}) \right] \\
&\leq \E\left[ (1-\alpha) G_{x_{t-1}}(x_{t-1}) + \alpha G_{x_{t-1}}(x_{\tau, H}) - F(x_{t-1}) \right] \\
&\leq \E\left[ \alpha \left( F(x_{\tau, H})  - F(x_{t-1}) \right) + \frac{\alpha \rho}{2} \| x_{\tau, H} - x_{t-1} \|^2 \right] \\
&\leq \alpha \E\left[  F(x_{\tau, H})  - F(x_{t-1}) \right] + \alpha \rho \| x_{\tau, H} - x_{\tau} \|^2 + \alpha \rho \| x_{\tau} - x_{t-1} \|^2  \\
&\leq \alpha \E\left[  F(x_{\tau, H})  - F(x_{t-1}) \right] + \alpha \rho \left(\gamma^2 H_{max}^2 \OM(V_2) +  \alpha^2\gamma^2 K^2 H_{max}^2 \OM(V_2)\right) \\
&\leq \alpha \E\left[  F(x_{\tau, H})  - F(x_{t-1}) \right] + \alpha \rho \left(\gamma^2 K^2 H_{max}^2 \OM(V_2) \right) \\
&\leq \alpha \E\left[  F(x_{\tau, H}) - F(x_{\tau}) + F(x_{\tau})  - F(x_{t-1}) \right] + \alpha \gamma^2 K^2 H_{max}^2 \OM(V_2).
\end{align*}


Using $L$-smoothness, we have
\begin{align*}
&F(x_\tau) - F(x_{t-1}) \\
&\leq \ip{\nabla F(x_{t-1})}{x_\tau - x_{t-1}} + \frac{L}{2} \|x_\tau - x_{t-1}\|^2 \\
&\leq \|\nabla F(x_{t-1})\| \|x_\tau - x_{t-1}\| + \frac{L}{2} \|x_\tau - x_{t-1}\|^2 \\
&\leq \sqrt{V_1} \alpha\gamma K H_{max} \OM(\sqrt{V_2}) + \frac{L}{2}  \alpha^2\gamma^2 K^2 H_{max}^2 \OM(V_2) \\
&\leq  \alpha\gamma K H_{max} \OM(\sqrt{V_1 V_2})   + \alpha^2\gamma^2 K^2 H_{max}^2 \OM(V_2).
\end{align*}

Thus, we have
\begin{align*}
&\E\left[ F(x_t) - F(x_{t-1}) \right] \\
&\leq - \alpha \gamma \epsilon \sum_{h = 0}^{H-1} \E \|\nabla F(x_{\tau, h})\|^2  + \alpha \gamma^2 \OM(\rho H_{max}^3 V_2) \\ 
&\quad + \alpha^2\gamma K H_{max} \OM(\sqrt{V_1 V_2}) + \alpha^3\gamma^2 K^2 H_{max}^2 \OM(V_2)\\
&\quad + \alpha \gamma^2 K^2 H_{max}^2 \OM(V_2).
\end{align*}

By rearranging the terms, we have
\begin{align*}
&\sum_{h = 0}^{H'_{t}-1} \E \|\nabla F(x_{\tau, h})\|^2 \\
&\leq \frac{\E\left[ F(x_{t-1}) - F(x_t) \right]}{\alpha \gamma \epsilon} + \frac{\gamma H_{max}^3}{\epsilon} \OM(V_2) \\
&\quad + \frac{\alpha K H_{max}}{\epsilon} \OM(\sqrt{V_1 V_2})  + \frac{\alpha^2\gamma K^2 H_{max}^2}{\epsilon} \OM(V_2) 
 +  \frac{\gamma K^2 H_{max}^2}{\epsilon} \OM(V_2),
\end{align*}
where $H'_{t}$ is the number of local iterations applied in the $t$th iteration.

By telescoping and taking total expectation, after $T$ global epochs, we have
\begin{align*}
&\min_{t = 0}^{T-1} \E\left[ \| \nabla F(x_t) \|^2 \right] \\
&\leq \frac{1}{\sum_{t = 1}^{T} H'_{t}} \sum_{t = 1}^{T} \sum_{h = 0}^{H'_{t}-1} \|\nabla F(x_{\tau, h})\|^2 \\
&\leq \frac{\E\left[ F(x_{0}) - F(x_T) \right]}{\alpha \gamma \epsilon T H_{min}}  + \frac{\gamma T H_{max}^3}{\epsilon T H_{min}} \OM(V_2) \\
&\quad + \frac{\alpha K T H_{max}}{\epsilon T H_{min}} \OM(\sqrt{V_1 V_2})  + \frac{\alpha^2\gamma K^2 T H_{max}^2}{\epsilon T H_{min}} \OM(V_2) 
 +  \frac{\gamma K^2 T H_{max}^2}{\epsilon T H_{min}} \OM(V_2) \\
&\leq \frac{\E\left[ F(x_{0}) - F(x_T) \right]}{\alpha \gamma \epsilon T H_{min}}  + \OM\left( \frac{\gamma H_{max}^3}{\epsilon H_{min}} \right)
 + \OM\left( \frac{\alpha K H_{max}}{\epsilon H_{min}} \right) \\
&\quad + \OM\left( \frac{\alpha^2\gamma K^2 H_{max}^2}{\epsilon H_{min}} \right) 
 + \OM\left( \frac{\gamma K^2 H_{max}^2}{\epsilon H_{min}} \right).
\end{align*}

Using $\delta = \frac{H_{max}}{H_{min}}$, and taking $\alpha = \frac{1}{\sqrt{H_{min}}}$, $\gamma = \frac{1}{\sqrt{T}}$, $T = H_{min}^5$, we have 
\begin{align*}
&\min_{t = 0}^{T-1} \E\left[ \| \nabla F(x_t) \|^2 \right] \\ 
&\leq \OM\left( \frac{1}{\epsilon H_{min}^3} \right)
 + \OM\left( \frac{\delta^3}{\epsilon \sqrt{H_{min}}} \right)
 + \OM\left( \frac{K \delta}{\epsilon \sqrt{H_{min}}} \right) 
 + \OM\left( \frac{K^2 \delta^2}{\epsilon \sqrt{H_{min}^5}} \right) 
 + \OM\left( \frac{K^2 \delta^2}{\epsilon \sqrt{H_{min}^3}} \right).
\end{align*}

\end{proof}

\section{Experiment details}



\subsection{NN architecture}
\label{sec:sup_exp}
In Table~\ref{tbl:cnn}, we show the detailed network structures of the CNN used in our experiments.
\begin{table}[htb!]
\centering
\caption{CNN Summary}
\label{tbl:cnn}
\begin{tabular}{|l|l|l|}
\hline 
Layer (type) & Parameters & Input Layer \\ \hline 
conv1(Convolution)& channels=64, kernel\_size=3, padding=1 &data \\ \hline 
activation1(Activation)& null &conv1 \\ \hline 
batchnorm1(BatchNorm)& null &activation1 \\ \hline 
conv2(Convolution)& channels=64, kernel\_size=3, padding=1 &batchnorm1 \\ \hline 
activation2(Activation)& null &conv2 \\ \hline 
batchnorm2(BatchNorm)& null &activation2 \\ \hline 
pooling1(MaxPooling)& pool\_size=2 &batchnorm2 \\ \hline 
dropout1(Dropout)& probability=0.25 &pooling1 \\ \hline 
conv3(Convolution)& channels=128, kernel\_size=3, padding=1 &dropout1 \\ \hline 
activation3(Activation)& null &conv3 \\ \hline 
batchnorm3(BatchNorm)& null &activation3 \\ \hline 
conv4(Convolution)& channels=128, kernel\_size=3, padding=1 &batchnorm3 \\ \hline 
activation4(Activation)& null &conv4 \\ \hline 
batchnorm4(BatchNorm)& null &activation4 \\ \hline 
pooling2(MaxPooling)& pool\_size=2 &batchnorm4 \\ \hline 
dropout2(Dropout)& probability=0.25 &pooling2 \\ \hline 
flatten1(Flatten)& null &dropout2 \\ \hline 
fc1(FullyConnected)& \#output=512 &flatten1 \\ \hline 
activation5(Activation)& null &fc1 \\ \hline 
dropout3(Dropout)& probability=0.25 &activation5 \\ \hline 
fc3(FullyConnected)& \#output=10 &dropout3 \\ \hline 
softmax(SoftmaxOutput)& null &fc3 \\ \hline   
\end{tabular} 
\end{table}


\end{document}